\pdfoutput=1
\RequirePackage{ifpdf}
\ifpdf 
\documentclass[pdftex]{sigma}
\else
\documentclass{sigma}
\fi

\numberwithin{equation}{section}

\newtheorem{Theorem}{Theorem}[section]
\newtheorem*{Theorem*}{Theorem}
\newtheorem{Corollary}[Theorem]{Corollary}
\newtheorem{Lemma}[Theorem]{Lemma}
\newtheorem{Proposition}[Theorem]{Proposition}
 { \theoremstyle{definition}
\newtheorem{Definition}[Theorem]{Definition}

\newtheorem{Example}[Theorem]{Example}
\newtheorem{Remark}[Theorem]{Remark} }

\begin{document}
\allowdisplaybreaks

\newcommand{\arXivNumber}{2209.01062}

\renewcommand{\PaperNumber}{092}

\FirstPageHeading

\ShortArticleName{Isomonodromic Deformations Along the Caustic of a Dubrovin--Frobenius Manifold}

\ArticleName{Isomonodromic Deformations Along the Caustic\\ of a Dubrovin--Frobenius Manifold}

\Author{Felipe REYES}

\AuthorNameForHeading{F.~Reyes}

\Address{SISSA, via Bonomea 265, Trieste, Italy}
\Email{\href{mailto:llopezre@sissa.it}{llopezre@sissa.it}}
\URLaddress{\url{https://lfelipe-lr.github.io/}}

\ArticleDates{Received May 03, 2023, in final form November 06, 2023; Published online November 16, 2023}

\Abstract{We study the family of ordinary differential equations associated to a Dubrovin--Frobenius manifold along its caustic. Upon just loosing an idempotent at the caustic and under a non-degeneracy condition, we write down a normal form for this family and prove that the corresponding fundamental matrix solutions are strongly isomonodromic. It is shown that the exponent of formal monodromy is related to the multiplication structure of the Dubrovin--Frobenius manifold along its caustic.}

\Keywords{Dubrovin--Frobenius manifolds; isomonodromic deformations; differential equations}

\Classification{53D45; 34M56}

\section{Introduction}
Dubrovin--Frobenius manifolds were invented by Boris Dubrovin to geometrize the study of certain $2D$ topological field theories~\cite{Dubrovin:1994hc,Dubrovin:1998fe}. The primary free energy $F$ of a family of such theories satisfies the so called {\it WDVV equations}. Given a quasi-homogeneous solution to these equations one constructs a Dubrovin--Frobenius manifold structure on the domain of definition~$M$ of the solution.

The first condition a Dubrovin--Frobenius manifold must satisfy is that the tangent sheaf $\mathcal{T}_{M}$ carries $\mathcal{O}_{M}$-bilinear multiplication $\circ\colon\mathcal{T}_{M}\times\mathcal{T}_{M}\rightarrow\mathcal{T}_{M}$, this multiplication is required to be unital, associative and commutative. The multiplication is required to satisfy an integrability condition (equation~(\ref{eq:Compatibility})); a manifold satisfying these conditions is called an \emph{$F$-manifold}. This integrability condition ensures that, the decomposition of each tangent space $T_{p}M$ into irreducible subalgebras extends to a decomposition of $M$ into irreducible $F$-manifolds. Next up in the definition comes the Euler vector field $E$, a global vector field required to satisfy $\mathcal{L}_{E}\circ=\circ$. Lastly one requires the existence of a flat metric $\eta$ satisfying some extra conditions (Definition~\ref{def:Dubrovin--FrobeniusManifold}).

As vector spaces, each tangent space $T_{p}M$ of a manifold is isomorphic to $\mathbb{C}^{n}$; on a Dubrovin--Frobenius manifold each tangent space is a $\mathbb{C}$-algebra and as a $\mathbb{C}$-algebra it may no longer be isomorphic to $\mathbb{C}^{n}$ (the multiplication in $\mathbb{C}^{n}$ is done entry by entry). If as algebras $T_{p}M\cong\mathbb{C}^{n}$ then the point $p$ is called semisimple. In this case, there exists a neighborhood $\bar{W}$ of $p$ such that all points in $\bar{W}$ are semisimple; in $\bar{W}$ there exist $n$-linearly independent vector fields $\bar{\pi}_{i}$ such that $\bar{\pi}_{i}\circ\bar{\pi}_{j}=\delta_{ij}\bar{\pi}_{i}$ and $[\bar{\pi}_{i},\bar{\pi}_{j}]=0$, these vectors are called \emph{orthogonal idempotents}. The points which are not semisimple form an hypersurface $K$ called the caustic which may be empty (see~\cite[Proposition 2.6]{HertlingFrobenius}). It is the purpose of this article to study the structure of a Dubrovin--Frobenius manifold in a neighborhood of a non semisimple point $p\in K$. In particular, we are interested in the restriction to the caustic of a family of differential equations associated to the Dubrovin--Frobenius manifold.

Closely related to a Dubrovin--Frobenius manifold is a family of meromorphic differential equations on $\mathbb{P}^{1}$ (on a small semisimple domain the structure of a Dubrovin--Frobenius manifold is equivalent to this family of equations, see~\cite{Dubrovin:1998fe}). Using the $\eta$-symmetric endomorphism $E\circ$ and the $\eta$-antisymmetric one $\mu:=\tfrac{2-d}{2}{\rm Id}-\nabla E$, for each point $p\in M$ the corresponding differential equation reads
\begin{equation*}
 \frac{{\rm d}y}{{\rm d}z}=\biggl(-(E\circ)_{p} +\frac{1}{z}\mu_{p}\biggr)y,
\end{equation*}
where $\nabla$ is the Levi-Civita connection of $\eta$. Each member of the family has a regular singularity at $z=0$ and an irregular at $z=\infty$. This family is induced by a vector bundle with a meromorphic flat connection over $\mathbb{P}^{1}\times M$ and as such one expects that the corresponding monodromy data are constant.

 At the regular singularity $z=0$, one gets that the Jordan form of the principal part (the endomorphism $\mu$) of the family of meromorphic differential equations does not depend on the point $p\in M$. As such, choosing a branch of the logarithm one writes a fundamental matrix solution for all $p\in M$ whose monodromy matrix is independent of $p\in M$.

 At the irregular singularity $z=\infty$, the normal form of the principal part (the endomorphism~$E\circ$) depends on the point $p\in M$, in particular, it is diagonal for all semisimple points and the basis that diagonalizes it, the orthogonal idempotents, ceases to exist in the caustic. In a neighborhood of a semisimple point, after choosing an ``admissible line'' in the $z$-plane, we can obtain holomorphic fundamental matrix solutions with prescribed asymptotic expansions in certain sectors and compute the Stokes matrices and the exponent of formal monodromy with respect to them; in a neighborhood of a semisimple point these data are also constant. Lastly, the central connection matrix relating the solutions found at $z=0$ and $z=\infty$ is also constant. In this way, in a neighborhood of a semisimple point, the family of meromorphic differential equations is strongly isomonodromic (i.e., the monodromy matrix, exponent of formal monodromy, Stokes matrices and the central connection matrix are constant~\cite{Dubrovin:1994hc,Dubrovin:1998fe}).

The fact that the monodromy data is constant on the semisimple domain depends only on the flatness and singularities of a connection on a vector bundle over $\mathbb{P}^{1}\times M$; the semisimplicity of the multiplication is nowhere used. Therefore, if one is able to define monodromy data for the non-semisimple points one expects it to be constant. It is the purpose of this paper to study this family of equations when only one idempotent is lost along the caustic, under one non-degeneracy assumption we define monodromy data and prove that it is constant.

Under our assumptions, when restricting to the caustic the principal part at $z=\infty$ will still be diagonalizable but it will have a repeated eigenvalue. Isomonodromic deformations of meromorphic differential equations on $\mathbb{P}^{1}$ with irregular singularities are studied in the seminal paper~\cite{a0b55245f91244e5bd4b591a56f5a1cb}. In their work, it is required that the principal parts at the irregular singularities have pairwise distinct eigenvalues. This condition is not satisfied by the system we are interested in. More recently, in~\cite{Guzzetti_2022} isomonodromic deformations of systems of the type we are interested are studied.

There are examples of Frobenius manifolds such that at the caustic more than one idempotent is lost (see~\cite{Basalaev_2014}). For our purpose, it is worth to have in mind the Frobenius manifolds coming from hypersurface singularities (see~\cite{1573950398998843136, Saito2008FROMPF}). In these examples it can be shown that at the caustic only one idempotent is lost.

\textbf{Main assumption and notations.} On the caustic $K$ the multiplication always has less than~$n$ orthogonal idempotents. Throughout this article we will work under the following assumptions:

\textbf{Assumption 1.} Generically along the caustic we have $n-1$ idempotents.

\textbf{Assumption 2.} The metric $\eta$ when pulled back to the caustic is non-degenerate so that this hypersurface has a well defined normal direction.

Under the first assumption for a point $p\in K$, we have that as algebras $T_{p}M\cong \mathbb{C}[z]/\bigl(z^{2}\bigr)\oplus\mathbb{C}^{n-2}$ correspondingly, the integrability condition~(\ref{eq:Compatibility}) tells us that, as germs of $F$-manifolds $(M,p)\cong F^{2}\times (A_{1})^{n-2}$ where $F^{2}$ is a two-dimensional $F$ manifold and the Euler vector field of $M$ decomposes as a sum of Euler vector fields on the corresponding manifolds (see~\cite[Theorem~2.11]{HertlingFrobenius}). The $F$-manifold $A_{1}$ is one-dimensional and we can choose a coordinate $u$ such that $e=\partial_{u}$ and $E=u\partial_{u}$. Germs of two-dimensional $F$-manifolds were classified in~\cite[Theorem~4.7]{HertlingFrobenius}. Each of these germs must be isomorphic to one of the $F$-manifolds $I_{2}(m)$ for some $m\in\mathbb{N}_{\geq 2}$. On these $F$-manifolds one can choose coordinates $(t,u_{2})$ such that $\pi_{2}:=\partial_{u_{2}}$ is the identity, $\partial_{t}\circ\partial_{t}=t^{m-2}\partial_{u_{2}}$ and the Euler vector field is
\begin{equation*}
 E=\tfrac{2}{m}t\partial_{t}+u_{2}\pi_{2}.
\end{equation*}
\indent Thus, under the first assumption, for a point $p\in K$ there exists a neighborhood $W$ of $p$ and a system of local coordinates $(t,u_{2},\dots,u_{n})$ on $W$ such that if we denote $\pi_{i}:=\partial_{u_{i}}$ for $i\geq 2$, then
\[
 \partial_{t}\circ\partial_{t}=t^{m-2}\pi_{2},\qquad
 \partial_{t}\circ\pi_{2}=\partial_{t},\qquad
 \pi_{2}\circ\pi_{2}=\pi_{2},
\]
and for $i,j\geq3$, we have
\[
 \pi_{2}\circ\pi_{i}=0,\qquad
 \partial_{t}\circ\pi_{i}=0,\qquad
 \pi_{i}\circ\pi_{j}=\delta_{ij}\pi_{i}.
\]
 On these coordinates, the caustic has the equation $t=0$ and the Euler vector field is
 \begin{equation*}
 E=\frac{2}{m}t\partial_{t}+\sum_{i=2}^{n}u_{i}\pi_{i}.
 \end{equation*}
 For a point $q\in W\setminus K$, there exists a neighborhood $\bar{W}$ of $q$ and canonical coordinates $(\bar{u}_{1},\dots,\bar{u}_{n})$ on $\bar{W}$ such that $\bar{\pi}_{i}\circ\bar{\pi}_{j}=\delta_{ij}\bar{\pi}_{i}$ where $\bar{\pi}_{i}:=\partial_{\bar{u}_{i}}$. The canonical coordinates $\bar{u}_{i}$ are the eigenvalues of the operator of multiplication by the Euler vector field hence, on the overlap $\bar{W}\cap W$ we have%
\begin{equation}
 \label{eq:ChangeCoordinates}
 \bar{u}_{1}=u_{2}+\tfrac{2}{m}t^{\frac{m}{2}},\qquad
 \bar{u}_{2}=u_{2}-\tfrac{2}{m}t^{\frac{m}{2}},\qquad
 \bar{u}_{i}=u_{i} \ \ \mathrm{for} \ i\geq 3.
\end{equation}
The corresponding basis of the tangent space are related by
\begin{equation}
 \label{eq:ChangeBasis}
 \partial_{t}=\big( \tfrac{m}{4}\big)^{\frac{m-2}{m}}(\bar{u}_{1}-\bar{u}_{2})^{\frac{m-2}{m}}(\bar{\pi}_{1}-\bar{\pi}_{2}),\qquad
 \pi_{2}=\bar{\pi}_{1}+\bar{\pi}_{2},\qquad
 \pi_{i}=\bar{\pi}_{i} \ \ \mathrm{for}\ i\geq 3.
\end{equation}
From now on, we will denote by $u_{i}$ and $\pi_{i}$ for $i\geq 3$ the coordinates and tangent vectors defined on $\bar{W}\cup W$.

\textbf{Organization.} In Section~\ref{sec:FManifolds}, we recall the definition of an $F$-manifold and of a Dubrovin--Frobenius manifold. We prove that under the main assumption, the caustic $K$ of a massive (generically semisimple) $F$-manifold $M$ is always a massive $F$-manifold (Proposition~\ref{prop:F-Caustic}). This result was previously known (see~\cite[Theorem~2.10 and Example 2.5]{Strachan_2004}) but we provide a different proof. Moreover, if $M$ has an Euler vector field then it is tangent to the caustic. This geometric fact has the consequence that multiplication by the Euler vector field $E$ is diagonalizable along the caustic; but the basis that diagonalizes it outside the caustic does not coincide with the one that diagonalizes it inside the caustic. As a corollary, we obtain that the caustic, with the induced structures, satisfies all the axioms of a Dubrovin--Frobenius manifold except for the flatness and in particular, the caustic of a three-dimensional Dubrovin--Frobenius manifold is always a Dubrovin--Frobenius manifold.

In Section~\ref{sec:DeformedConnection}, we recall the definition of the deformed connection of a Dubrovin--Frobenius manifold and of the associated family of meromorphic ordinary differential equations on $\mathbb{P}^{1}$ that it determines. We then pullback this family to the caustic and, if the metric $\eta$ is non-degenerate along the caustic, we write down this family in a convenient basis.

In Section~\ref{sec:Monodromy}, we recall the monodromy data at $z=0$. As shown in~\cite{Dubrovin:1998fe}, these data does not depend on the point $p\in M$.

Outside the caustic the exponent of formal monodromy is identically zero, but this is no longer true inside the caustic. In Section~\ref{sec:ExponentOfFormalMonodromy}, we write down formal solutions for the family of differential equations, compute the exponent of formal monodromy, show that it is constant and that it is related to the decomposition $(M,p)\cong I_{2}(m)\times (A_{1})^{n-2}$ (Proposition~\ref{prop:ExponentMultiplication}).

In Section~\ref{sec:StokesMatrices}, we invoke Sibuya's theorem to find holomorphic solutions having the asymptotics of the formal solutions found in the previous section. Using them, we define the Stokes matrices and show that they are constant. We also prove that the connection matrix relating the particular solutions found at $z=0$ and $z=\infty$ is constant. Combining these facts, we obtain that the pullback of the family of differential equation is strongly isomonodromic (Theorem~\ref{thm:Isomonodromy}).

 Section~\ref{sec:Examples} provides some three-dimensional examples.

It is natural to ask if Proposition~\ref{prop:ExponentMultiplication} and Theorem~\ref{thm:Isomonodromy} can be generalized when one looses more than one idempotent. As we will see, the classification of germs of $2$-dimensional $F$-manifolds will play a crucial role. Recently significant progress has been made in the classification of germs of $3$-dimensional $F$-manifolds; but as seen in~\cite{Basalaev_2021} this classification is much more vast and complicated than the two-dimensional case.

\section[The caustic of an F-manifold]{The caustic of an $\boldsymbol{F}$-manifold}
\label{sec:FManifolds}

In this section, we prove that under the main assumption the caustic of a massive $F$-manifold is again a massive $F$-manifold. If the starting $F$-manifold was a Dubrovin--Frobenius manifold then, with the induced structures, the caustic satisfies all the axioms of a Dubrovin--Frobenius manifold except for the flatness of the metric. For $3$-dimensional Dubrovin--Frobenius manifolds, the caustic is always a Dubrovin--Frobenius manifold.
\begin{Definition}
 An \emph{$F$-manifold} is a triple $(M,\circ,e)$ where $M$ is a complex manifold of dimension~$n$, $\circ\colon\mathcal{T}_{M}\otimes\mathcal{T}_{M}\rightarrow\mathcal{T}_{M}$ is a commutative, associative and $\mathcal{O}_{M}$-bilinear multiplication, $e$ is a~global unit field for $\circ$ and the multiplication satisfies for any two local vector fields $X$, $Y$
 \begin{equation}
 \label{eq:Compatibility}
 \mathcal{L}_{X\circ Y} ( \circ ) = X\circ\mathcal{L}_{ Y} ( \circ ) + Y\circ\mathcal{L}_{X} ( \circ ).
 \end{equation}
 An Euler vector field for the $F$-manifold $M$ is a global vector field $E$ such that
 \begin{equation}
 \label{eq:Euler}
 \mathcal{L}_{E}( \circ )= \circ.
 \end{equation}
 A point $p\in M$ is called \emph{semisimple} if the tangent space $T_{p}M$ has no nilpotents. In this case, it can be shown that $T_{p}M$ decomposes as a sum of one-dimensional algebras $\sum_{i=1}^{n}\mathbb{C}\cdot \bar{\pi}_{i}$ where the vectors $\bar{\pi}_{i}$ satisfy $\bar{\pi}_{i}\circ\bar{\pi}_{j}=\delta_{ij}\bar{\pi}_{i}$ and $[\bar{\pi}_{i},\bar{\pi}_{j}]=0$; these vectors are called \emph{orthogonal idempotents}. The \emph{caustic} $K\subset M$ is the set of points which are not semisimple. In~\cite[Proposition~2.6]{HertlingFrobenius}, it is shown that the caustic is either empty or an hypersurface.
 Here we prove the following proposition (see~\cite{Strachan_2004}).
\end{Definition}
\begin{Proposition}
 \label{prop:F-Caustic}
 Let $(M,\circ,e)$ be a massive $F$-manifold of dimension $n$ and let $K\neq\varnothing$ be its caustic. Denote by $i\colon K\rightarrow M$ the inclusion. Suppose there exists a codimension~$1$ subvariety $($in $K)$ $\tilde{K}\subset K$ such that there exist $n-1$ vector fields $\pi_{2},\dots,\pi_{n}\in\Gamma\bigl(K\setminus\tilde{K}, i^{*}\mathcal{T}_{M}\bigr)$ such that $\pi_{i}\circ\pi_{j}=\delta_{ij}\pi_{i}$. Then the caustic $K$ is a massive $F$-manifold of dimension $n-1$ and the vectors~$\pi_{i}$ are tangent to it. Moreover, if $E$ is an Euler vector field for the $F$-manifold~$M$ then~$E$ is tangent to $K$ and it is an Euler vector field for it. The endomorphism $i^{*}E\circ\colon i^{*}\mathcal{T}_{M}\rightarrow i^{*}\mathcal{T}_{M}$ is diagonalizable along $K\setminus\tilde{K}$.
\end{Proposition}
\begin{proof}
 The existence of $n-1$ orthogonal idempotents $\pi_{i}$ tells us that as an algebra $T_{p}M$ decomposes as $V\oplus\bigl(\bigoplus_{i=3}^{n}\mathbb{C}\cdot\pi_{i}\bigr)$, where $V$ is a $2$-dimensional algebra and~$\pi_{2}$ is the unit on~$V$. We will use two results of~\cite{HertlingFrobenius}. According to~\cite[Theorem~2.11]{HertlingFrobenius}, the above decomposition extends to a decomposition of the germ of the $F$-manifold $M$ at the point $p$ and the Euler vector field decomposes as a sum of Euler vector fields of the corresponding $F$-manifolds. In this case, the decomposition is $F^{2}\times\bigl(\prod_{i=1}^{n-2}A_{1}\bigr) $ and the Euler vector field decomposes as
 \begin{equation*}
 E= v+\sum_{i=3}^{n}u_{i}\pi_{i},
 \end{equation*}
 with $v\in V$. The second result we will use is the classification of two-dimensional germs of $F$-manifolds (see~\cite[Example~2.12\,(iv) and Theorem~4.7\,(a)]{HertlingFrobenius}).
 There it is proven that $F^{2}$ must be isomorphic to one of the germs $I_{2}(m)$. These germs admit local coordinates $(t,u_{2})$ (here we use a different notation from the one on~\cite{HertlingFrobenius}) such that $\partial_{u_{2}}=\pi_{2}$ is the identity on $\mathcal{T}_{I_{2}(m)}$, $\partial_{t}\circ\partial_{t}=t^{m-2}\pi_{2}$ and the Euler vector field is $v=\tfrac{2}{m}t\partial_{t}+u_{2}\pi_{2}$. Since on the hypersurface $t=0$ the vector $\partial_{t}$ is nilpotent this hypersurface is contained in $K\setminus\tilde{K}$. The vectors tangent to this hypersurface are $\pi_{2},\dots,\pi_{n}$. The Euler vector field is
 \begin{equation*}
 E=\frac{2}{m}t\partial_{t}+\sum_{i=2}^{n}u_{i}\pi_{i}.
 \end{equation*}
 Along the caustic, the basis $\partial_{t},\pi_{2},\dots,\pi_{n}$ of $i^{*}\mathcal{T}_{M}$ diagonalizes the endomorphism $E\circ$ and the eigenvalues are $u_{3},\dots,u_{n}$ of multiplicity one and $u_{2}$ which has multiplicity two.
\end{proof}

\begin{Definition}
 \label{def:Dubrovin--FrobeniusManifold}
 A \emph{Dubrovin--Frobenius manifold} is a tuple $(M,\circ,e,E,\eta)$, where $(M,\circ,e)$ is an $F$-manifold with Euler vector field $E$ and $\eta$ is a metric on $M$ satisfying
 \begin{enumerate}\itemsep=0pt
 \item For any vector fields $X$, $Y$, $Z$ we have $\eta(X\circ Y,Z)=\eta(X,Y\circ Z)$.
 \item The unit $e$ is flat, namely $\nabla e=0$ where $\nabla$ is the Levi-Civita connection of $\eta$.
 \item The Euler vector field satisfies $\mathcal{L}_{E}\eta=(2-d)\eta$.
 \item The metric $\eta$ is flat.
 \end{enumerate}
 \end{Definition}
\begin{Corollary}
 Let $(M,\circ,e,E,\eta)$ be a Dubrovin--Frobenius manifold and suppose that the caustic $K$ satisfies the hypothesis of Proposition~$\ref{prop:F-Caustic}$. Then $(K,\circ,e,E,i^{*}\eta)$ satisfies all the axioms of Dubrovin--Frobenius manifold except possibly for the flatness of $i^{*}\eta$. Moreover, if $M$ is $3$-dimensional then $(K,\circ,e,E,i^{*}\eta)$ is a Dubrovin--Frobenius manifold.
\end{Corollary}
\begin{proof}
 The only thing that needs to be proven is the statement about the $3$-dimensional Dubrovin--Frobenius manifold. Let $g=i^{*}\eta$ and let $\tilde{\nabla}$ denote the Levi-Civita connection of~$g$. By hypothesis, $\nabla e=0$ so projecting to the tangent space of the caustic gives $\tilde{\nabla} e=0$. Using this, we get $\mathcal{L}_{e}g=\tilde{\nabla}_{e}g=0$. Call $\partial_{1}=e$ and pick a vector field $\partial_{2}$ such that $[\partial_{1},\partial_{2}]=0$. Then $\mathcal{L}_{e}g=0$ implies $\partial_{1}g_{ij}=0$ so that the components of the metric in this basis are constant in the direction of the unit vector field. Since the Christoffel symbols are functions of the metric and its derivatives, they are also constant along the unit vector field. Now $\big[\tilde{\nabla}_{\partial_{1}},\tilde{\nabla}_{\partial_{2}}\big]\partial_{1}=0$ because $\tilde{\nabla}_{e}=0$. Finally,
 \begin{equation*}
 \big[\tilde{\nabla}_{\partial_{1}},\tilde{\nabla}_{\partial_{2}}\big]\partial_{2}=\Gamma_{22}^{1}
 \nabla_{\partial_{1}}\partial_{1}+\Gamma_{22}^{2}\nabla_{\partial_{1}}\partial_{2}=0.\tag*{\qed}
 \end{equation*}\renewcommand{\qed}{}
 \end{proof}

\begin{Example}
 Let us consider the Dubrovin--Frobenius manifold $M$ associated with the singularity $A_{n}$. This manifold consists of the polynomials of the form
 \begin{equation*}
 F(a;z)=z^{n+1}+a_{n-1}z^{n-1}+\cdots + a_{1}z + a_{0},
 \end{equation*}
 where $a=(a_{0},\dots,a_{n-1})\in\mathbb{C}^{n}$. This manifold is an affine space modeled on the vector space of polynomials of degree at most $n-1$. This means that we can identify the tangent space to any point $a\in M$ with the space of polynomials of degree at most $n-1$. Given two polynomials $f,g\in T_{a}M$ the multiplication is defined by
 \begin{equation*}
 f\circ g:= fg \mod \frac{\partial F}{\partial z}\bigg|_{a}.
 \end{equation*}
 If we write $\frac{\partial F}{\partial z}=(n+1)\prod_{i=1}^{n}(z-\alpha_{i})$, then one can easily check that the polynomials
 \begin{equation*}
 e_{i}:=\frac{1}{z-\alpha_{i}}\frac{\partial F}{\partial z}
 \end{equation*}
satisfy $e_{i}\circ e_{j}=\delta_{ij}\lambda_{i}e_{i}$ with $\lambda_{i}=e_{i}(\alpha_{i})$ and therefore they are multiples of the orthogonal idempotents. Hence the caustic $K$ consist of the points $a$ such that the polynomial $\frac{\partial F}{\partial z}$ has a~double root. The set of points where $\frac{\partial F}{\partial z}$ has only a double root and all other roots simple is an open set inside the caustic. In this open set the polynomials $e_{i}$, with $\alpha_{i}$ a simple root, still are multiples of the orthogonal idempotents $\pi_{i}$; we have $n-2$ of them, say $\pi_{3},\dots,\pi_{n}$. But we have another orthogonal idempotent given by $e-\pi_{3}-\cdots-\pi_{n}$. By Proposition~\ref{prop:F-Caustic}, the caustic is a massive $F$-manifold. Note that we can apply the proposition again, indeed, the caustic contains the locus of points $\tilde{K}$ such that the polynomial $\tfrac{\partial F}{\partial z}$ has a triple root and all the other roots simple. The same argument as before shows that along $\tilde{K}$ we have $n-2$ orthogonal idempotents. Continuing in this way, we arrive at a $2$-dimensional $F$-manifold, the locus of points where $\tfrac{\partial F}{\partial z}$ has a root of multiplicity $n-1$ and a simple root. By the corollary, this surface is a Dubrovin--Frobenius manifold. Dubrovin--Frobenius surfaces are classified by their charge. Since the charge of the $A_{n}$ Dubrovin--Frobenius manifold is $d=\frac{n-1}{n+1}$ (see~\cite{Dubrovin:1993nt}), the corresponding surface is isomorphic to $I_{2}(n+1)$.
\end{Example}
On the future, we will use the following statement.

\begin{Lemma}
 \label{prop:Closedness}
 Let $(M,\circ,e,E,\eta)$ be a Dubrovin--Frobenius manifold. Then the $1$-form $\eta(e,-)$ is closed.
\end{Lemma}
\begin{proof}
 By torsion freeness of $\nabla$, we have
 \begin{align*}
 {\rm d}(\eta(e,-))(u,v)&=u\eta(e,v)-v\eta(e,u)-\eta(e,[u,v])\\
 &=u\eta(e,v)-v\eta(e,u)-\eta(e,\nabla_{u}v-\nabla_{v}u).
 \end{align*}
 By compatibility of the metric, we get
 \begin{equation*}
 {\rm d}(\eta(e,-))(u,v)=\eta(\nabla_{u}e,v)-\eta(u,\nabla_{v}e).
 \end{equation*}
 So by flatness of $e$, we get the result.
\end{proof}

\section{The deformed connection}
\label{sec:DeformedConnection}
In this section, we describe the deformed connection of a Dubrovin--Frobenius manifold $M$. This consists of a family of connections parametrized by $z\in\mathbb{C}$ and at $z=0$, we recover the Levi-Civita connection of $\eta$. Thanks to the properties of a Dubrovin--Frobenius manifold, for any $z\in\mathbb{C}$ the corresponding connection is flat. Moreover, this connection can be extended to a flat connection with singularities on a vector bundle over $\mathbb{P}^{1}\times M$. By considering only the derivative in direction of $\mathbb{P}^{1}$, every Dubrovin--Frobenius manifold determines a family of ordinary differential equations on $\mathbb{P}^{1}$, this family is parametrized by the points of the Dubrovin--Frobenius manifold. By pulling this connection to $\mathbb{P}\times K$, we obtain a new family of ordinary differential equations parametrized by the caustic.

\begin{Definition}
 Let $(M,\circ,e,E,\eta)$ be a Dubrovin--Frobenius manifold, let $\nabla$ be the Levi-Civita connection of $\eta$ and let $z$ be a global coordinate on $\mathbb{C}$. The \emph{deformed connection} is a $1$-parameter family of connections on $\mathcal{T}_{M}$. For $z\in\mathbb{C}$, it is defined as
 \begin{equation*}
 \bar{\nabla}:=\nabla + z\circ.
 \end{equation*}
\end{Definition}
Thanks to the commutativity of $\circ$, the deformed connection is torsionless. Moreover, the flatness of $\nabla$, properties~(1) and~(2) of Definition~\ref{def:Dubrovin--FrobeniusManifold} and the associativity of $\circ$, imply that the deformed connection is flat for every~$z\in\mathbb{C}$.

Consider the projections $\pi\colon\mathbb{P}^{1}\times M\rightarrow M$ and $\pi_{1}\colon\mathbb{P}^{1}\times M\rightarrow\mathbb{P}^{1}$. We now extend the deformed connection to a connection with singularities on the bundle $\pi^{*}\mathcal{T}_{M}$. First note that thanks to property~(4) of Definition~\ref{def:Dubrovin--FrobeniusManifold}, the endomorphism of $\mathcal{T}_{M}$ defined by
\begin{equation*}
 \mu:=\frac{2-d}{2}-\nabla E
\end{equation*}
is $\eta$-antisymmetric. The $\mathcal{O}_{M}$-linear tensors $\circ$ and $\mu$ on $\mathcal{T}_{M}$ induce $\mathcal{O}_{\mathbb{P}^{1}\times M}$-linear tensors on $\pi^{*}\mathcal{T}_{M}$. Abusing notation we will denote them by the same symbols. Recall that $\mathcal{T}_{\mathbb{P}^{1}\times M}\cong \pi^{*}_{1}\mathcal{T}_{\mathbb{P}^{1}}\oplus \pi^{*}\mathcal{T}_{M}$. The connection $\bar{\nabla}$ on $\pi^{*}\mathcal{T}_{M}$ is defined as
\begin{equation*}
 \bar{\nabla}_{u}v:=\pi^{*}\left(\bar{\nabla}\right)_{u}v =\nabla_{u}v+zu\circ v
 \end{equation*}
 and
 \begin{equation*}
 \bar{\nabla}_{\partial_{z}}v:=\frac{\partial v}{\partial z}+E\circ v-\frac{1}{z}\mu v,
 \end{equation*}
where $u,v\in\pi^{*}\mathcal{T}_{ M}$ and $\partial_{z}$ is the vector field associated with the global coordinate $z$. The equality~(\ref{eq:Euler}) guarantees that $\bar{\nabla}$ is a flat connection on $\pi^{*}\mathcal{T}_{M}$. This means that for any point $(z,p)\in\mathbb{P}^{1}\times M$ we can find $n$ linearly independent sections $v_{i}\in\pi^{*}\mathcal{T}_{M}$ that satisfy $\bar{\nabla}v_{i}=0$. In particular, fixing a basis and putting the components of these $n$ sections as columns of a~matrix~$Y$, we get that, for each $p\in M$, $Y$ satisfies the ordinary differential equation
\begin{equation}
 \label{eq:TheDifferentialEquation}
 \frac{{\rm d}Y}{{\rm d}z}=\biggl(\frac{1}{z}\mu - E\circ \biggr)Y.
 \end{equation}
 Now take a semisimple point $q\in M$ and let
 \begin{equation*}
 \bar{f}_{i}:=\frac{\bar{\pi}_{i}}{|\bar{\pi}_{i}|}, \qquad i=1,\dots,n,
 \end{equation*}
 denote the normalized orthogonal idempotents on $M\setminus K$, where the length $|\cdot|$ is computed using the metric $\eta$. We have that $E\circ \bar{f}_{i}=\bar{u}_{i}\bar{f}_{i}$ so that the matrix $\bar{U}$ representing $E\circ$ is diagonal and the matrix $\bar{V}$ representing $\mu$ is antisymmetric. In this basis, the system~(\ref{eq:TheDifferentialEquation}) is written as
 \begin{equation*}
 \frac{{\rm d}Y}{{\rm d}z}=\biggl( \frac{1}{z}\bar{V}-\bar{U}\biggr)Y.
 \end{equation*}
 In~\cite{Dubrovin:1994hc}, it is shown that around a semisimple point $q\in M$ such that $E\circ$ has different eigenvalues the family of ordinary differential equations~(\ref{eq:TheDifferentialEquation}) is isomonodromic. Moreover, in~\cite{MR4094756} this result was extended to points where $E\circ$ has repeated eigenvalues but the multiplication remains semisimple. In this article, we show that under the conditions of Proposition~\ref{prop:F-Caustic} the family~(\ref{eq:TheDifferentialEquation}) remains isomonodromic when pulled back to the caustic.

 Consider the inclusion $j\colon\mathbb{P}^{1}\times K\rightarrow\mathbb{P}^{1}\times M$ ($j={\rm id}_{\mathbb{P}^{1}}\times i$) and the vector bundle $(\pi\circ j)^{*}\mathcal{T}_{M}$ over $\mathbb{P}^{1}\times K$. This vector bundle has a flat connection $j^{*}\bar{\nabla}$. Moreover, if the metric $i^{*}\eta$ is non-degenerate on every point of the caustic $K$ we can find a unitary normal vector $N$ to $K$. Consider the vectors $\partial_{t},\pi_{2},\dots,\pi_{n}$ provided by Hertling's decomposition (see equalities~(\ref{eq:ChangeCoordinates}) and~(\ref{eq:ChangeBasis})); by the compatibility of the multiplication $\circ$ and the metric $\eta$ we must have that $N$ is a linear combination of the vectors $\partial_{t}$ and $\pi_{2}$. As shown in Proposition~\ref{prop:F-Caustic}, in the subspace generated by these two vectors $E\circ$ acts by multiplication by $u_{2}$. Therefore, on the basis
 \begin{equation*}
 N, \quad f_{2}:=\frac{\pi_{2}}{|\pi_{2}|},\quad f_{i}:=\bar{f_{i}}, \quad i=3,\dots,n,
 \end{equation*}
 the matrix $U$ representing $E\circ$ is diagonal and the matrix $V$ representing $\mu$ is antisymmetric. System~(\ref{eq:TheDifferentialEquation}) takes the form
 \begin{equation*}
 \frac{{\rm d}Y}{{\rm d}z}=\biggl( \frac{1}{z}V-U\biggr)Y.
 \end{equation*}
 Hence we see that pulling back the family~(\ref{eq:TheDifferentialEquation}) to the caustic, we get a family of the same kind but with one parameter less and the operator $E\circ$ has an eigenvalue of multiplicity $2$, namely $u_{2}$.

 For later use, let us write down the connection matrices of the flat connection $j^{*}\bar{\nabla}$ on the vector bundle $(\pi\circ j)^{*}\mathcal{T}_{M}$ over $\mathbb{P}^{1}\times K$. We will use the frame $\partial_{z},\pi_{2},\dots,\pi_{n}$ of $\mathcal{T}_{\mathbb{P}^{1}\times K}$ and the frame $N,f_{2},\dots,f_{n}$ of $(\pi\circ j)^{*}\mathcal{T}_{M}$. By the above discussion, the $z$-component is
 \begin{equation*}
 \omega_{z}=U-\frac{1}{z}V.
 \end{equation*}
 Now $\bar{\nabla}_{\pi_{2}}=\nabla_{\pi_{2}}+z\pi_{2}\circ$ and since $\pi_{2}\circ$ is the identity on the subspace generated by $N$ and $\pi_{2}$, and zero on the subspace generated by $f_{3},\dots,f_{n}$, we have
 \begin{equation*}
 \bar{\omega}_{2}=\omega_{2}+zE_{2},
 \end{equation*}
 where $\omega_{2}$ is the connection matrix of the flat connection $i^{*}\nabla$ and $E_{2}$ has a $2\times 2$ identity matrix on the highest leftmost block and all the other entries are zero, i.e.,
 \begin{equation*}
 (E_{2})^{\alpha}_{\beta}=\delta^{\alpha}_{1}\delta^{1}_{\beta}+\delta^{\alpha}_{2}\delta^{2}_{\beta}.
 \end{equation*}
 Analogously, for $i>2$ we have
 \begin{equation*}
 \bar{\omega}_{i}=\omega_{i}+zE_{i},
 \end{equation*}
 where the matrices $\omega_{i}$ are the connection matrices of the flat connection $i^{*}\nabla$ and $(E_{i})^{\alpha}_{\beta}=\delta^{\alpha}_{i}\delta^{i}_{\beta}$. Notice that since the first two eigenvalues of the matrix $U$ are $u_{2}$ we have ${\rm d}U=\sum_{i=2}^{n}E_{i}{\rm d}u_{i}$ and hence the connection form of the connection $j^{*}\bar{\nabla}$ can be written as
 \begin{equation}
 \label{eq:ConnectionForm}
 \bar{\omega}=z{\rm d}U+\sum_{i=2}^{n}\omega_{i}{\rm d}u_{i}.
 \end{equation}
 The following lemma will be useful for some computations.
 \begin{Lemma}
 We have the following identities:
 \begin{align}
 & [E_{i},\omega_{j}]=[E_{j},\omega_{i}],\label{eq:UsefulIdentities1} \\
 & \frac{\partial V}{\partial u_{i}}=[V,\omega_{i}], \label{eq:UsefulIdentities2} \\
 & \left[U,\omega_{i}\right]=-[E_{i},V]. \label{eq:UsefulIdentities3}
 \end{align}
 \end{Lemma}
 \begin{proof}
 Since the connection $\bar{\nabla}$ on $\pi^{*}\mathcal{T}_{M}$ is flat, the connection $j^{*}\nabla$ on $(\pi\circ j)^{*}\mathcal{T}_{M}$ is also flat. Flatness of $j^{*}\nabla$ in the plane generated by $\partial_{i}$ and $\pi_{j}$ gives~(\ref{eq:UsefulIdentities1}) and flatness in the plane generated by $\partial_{z}$ and $\pi_{i}$ gives~(\ref{eq:UsefulIdentities2}) and~(\ref{eq:UsefulIdentities3}).
 \end{proof}

 Recall that for any point of the caustic $(M,p)\cong I_{2}(m)\times(A_{1})^{n-2}$. We now introduce a~connection on the subbundle $i^{*}\mathcal{T}_{I_{2}(m)}$ of $i^{*}\mathcal{T}_{M}$. This bundle is generated by the vectors~$N$,~$f_{2}$. The following connection will be useful when studying isomonodromic deformations of equation~(\ref{eq:TheDifferentialEquation}), it is defined as
 \begin{equation*}
 ({\rm id}\otimes\pi_{2}\circ)i^{*}\nabla\colon \ \Omega^{0}_{K}\otimes i^{*}\mathcal{T}_{I_{2(m)}}\rightarrow\Omega^{1}_{K}\otimes i^{*}\mathcal{T}_{I_{2(m)}},
 \end{equation*}
 where ${\rm id}$ is the identity on $\Omega^{1}_{K}$. Since $\pi_{2}\circ$ is the identity on the subbundle generated by the vectors $N$, $f_{2}$, the past expression does define a connection. Indeed, $\mathbb{C}$-linearity is clear and if $h\in\mathcal{O}_{K}$ is a holomorphic function and $v\in \langle N,f_{2}\rangle$, we have
 \begin{align*}
 ({\rm id}\otimes\pi_{2}\circ)i^{*}\nabla hv& =({\rm id}\otimes\pi_{2}\circ)({\rm d}h\otimes v+hi^{*}\nabla v)\\
 & ={\rm d}h\otimes\pi_{2}\circ v+h({\rm id}\otimes\pi_{2}\circ)i^{*}\nabla v\\
 & ={\rm d}h\otimes v+h({\rm id}\otimes\pi_{2}\circ)i^{*}\nabla v
 \end{align*}
 because $({\rm id}\otimes\pi_{2}\circ)$ is $\mathcal{O}_{K}$-linear. Let us compute the connection matrices of this connection. We have
 \begin{align*}
 & i^{*}\nabla N=\sum_{i=2}^{n}\Biggl((\omega_{i})_{1}^{1}{\rm d}u_{i}\otimes N+\sum_{s=2}^{n}(\omega_{i})^{s}_{1}{\rm d}u_{i}\otimes f_{s}\Biggr),\\
 & i^{*}\nabla f_{2}=\sum_{i=2}^{n}\Biggl((\omega_{i})_{2}^{1}{\rm d}u_{i}\otimes N+\sum_{s=2}^{n}(\omega_{i})^{s}_{2}{\rm d}u_{i}\otimes f_{s}\Biggr),
 \end{align*}
 so that
 \begin{align*}
 & ({\rm id}\otimes\pi_{2}\circ)i^{*}\nabla N=\sum_{i=2}^{n}(\omega_{i})^{1}_{1}{\rm d}u_{i}\otimes N+(\omega_{i})^{2}_{1}{\rm d}u_{i}\otimes f_{2},\\
 & ({\rm id}\otimes\pi_{2}\circ)i^{*}\nabla f_{2}=\sum_{i=2}^{n}(\omega_{i})^{1}_{2}{\rm d}u_{i}\otimes N+(\omega_{i})^{2}_{2}{\rm d}u_{i}\otimes f_{2}.
 \end{align*}
 Therefore, for $i=2,\dots,n$,
 \begin{equation*}
 \begin{aligned}
 ({\rm id}\otimes\pi_{2}\circ)i^{*}\nabla_{\pi_{i}}N&=(\omega_{i})^{1}_{1}N+(\omega_{i})^{2}_{1}f_{2},\\
 ({\rm id}\otimes\pi_{2}\circ)i^{*}\nabla_{\pi_{i}}f_{2}&=(\omega_{i})^{1}_{2}N+(\omega_{i})^{2}_{2}f_{2}.
 \end{aligned}
 \end{equation*}
 Hence the connection matrices $_{2}\omega_{i}$ of the connection $({\rm id}\otimes\pi_{2}\circ)i^{*}\nabla$ on $i^{*}\mathcal{T}_{I_{2}(m)}$ are
 \begin{equation*}
 _{2}\omega_{i}=E_{2}\omega_{i}E_{2}.
 \end{equation*}
 In other words, the connection matrices $_{2}\omega_{i}$ of the connection $({\rm id}\otimes\pi_{2}\circ)i^{*}\nabla$ are the highest leftmost $2\times 2$ block of the connection matrices $\omega_{i}$ of the connection $i^{*}\nabla$ on the vector bundle~$i^{*}\mathcal{T}_{M}$.
 \begin{Proposition}
 \label{prop:FlatnessDiagonalConnection}
 The connection $({\rm id}\otimes\pi_{2}\circ)i^{*}\nabla$ is flat.
 \end{Proposition}
 \begin{proof}
 The connection $\nabla$ on $\mathcal{T}_{M}$ is flat. Hence the connection $i^{*}\nabla$ on $i^{*}\mathcal{T}_{M}$ is flat, in particular for $\alpha,\beta\in\{ 1,2 \}$ and $i,j\in\{ 2,\dots,n \}$, $i\neq j$, we have
 \begin{equation*}
 0=\frac{\partial (\omega_{i})^{\alpha}_{\beta}}{\partial u_{j}}-\frac{\partial (\omega_{j})^{\alpha}_{\beta}}{\partial u_{i}}-[\omega_{i},\omega_{j}]^{\alpha}_{\beta}.
 \end{equation*}
 Let us compute
 \begin{equation*}
 [\omega_{i},\omega_{j}]^{\alpha}_{\beta}=\sum_{s=1,2}\bigl((\omega_{i})^{\alpha}_{s}(\omega_{j})^{s}_{\beta}-(\omega_{j})^{\alpha}_{s}(\omega_{i})^{s}_{\beta}\bigr)
 +\sum_{s=3}^{n}\bigl((\omega_{i})^{\alpha}_{s}(\omega_{j})^{s}_{\beta}-(\omega_{j})^{\alpha}_{s}(\omega_{i})^{s}_{\beta}\bigr).
 \end{equation*}
 Suppose $i=2$ and $j>2$, for $s>2$ from equation~(\ref{eq:UsefulIdentities1}), we obtain $(\omega_{j})^{\alpha}_{s}=-\delta_{js}(\omega_{2})^{\alpha}_{s}$ and $(\omega_{j})^{s}_{\beta}=-\delta_{js}(\omega_{2})^{s}_{\beta}$. Hence,
 \begin{equation*}
 \sum_{s=3}^{n}(\omega_{2})^{\alpha}_{s}(\omega_{j})^{s}_{\beta}-(\omega_{j})^{\alpha}_{s}(\omega_{2})^{s}_{\beta}=0.
 \end{equation*}
 Now suppose $i,j>2$, again from equation~(\ref{eq:UsefulIdentities1}) and $s>2$, we have $(\omega_{i})^{\alpha}_{s}=\delta_{is}(\omega_{s})^{\alpha}_{s}$ and $(\omega_{j})^{s}_{\beta}=\delta_{js}(\omega_{s})^{s}_{\beta}$ so again the above sum vanishes.

 Therefore, for $\alpha,\beta\in\{ 1,2 \}$ and $i,j\in\{ 2,\dots,n \}$, $i\neq j$, we obtain
 \begin{equation*}
 0=\frac{\partial (\omega_{i})^{\alpha}_{\beta}}{\partial u_{j}}-\frac{\partial (\omega_{j})^{\alpha}_{\beta}}{\partial u_{i}}-\sum_{s=1,2}(\omega_{i})^{\alpha}_{s}(\omega_{j})^{s}_{\beta}-(\omega_{j})^{\alpha}_{s}(\omega_{i})^{s}_{\beta}.
 \end{equation*}
 But this last expression is nothing else than the curvature tensor of connection $({\rm id}\otimes\pi_{2}\circ)i^{*}\nabla$.
 \end{proof}

 \begin{Remark}
 In the proof of the previous proposition, we only used equality~(\ref{eq:UsefulIdentities1}). This equality also holds true for the so called flat $F$-manifolds (see~\cite{https://doi.org/10.48550/arxiv.2001.05599,https://doi.org/10.48550/arxiv.2104.09380}), these $F$-manifolds are equipped with a flat connection $\nabla$ on the tangent bundle such that the deformed connection $\bar{\nabla}=\nabla+z\circ$ is flat for every $z\in\mathbb{C}$.
 \end{Remark}

\section[Monodromy data at z=0]{Monodromy data at $\boldsymbol{z=0}$}
 \label{sec:Monodromy}

 It this section, we describe the monodromy data of equation~(\ref{eq:TheDifferentialEquation}) at $z=0$. We refer to~\cite{Dubrovin:1994hc}, where it is shown that the monodromy matrix of the Levelt form solution is constant.

 The singularity of equation~(\ref{eq:TheDifferentialEquation}) at $z=0$ is Fuchsian, and therefore there exists a holomorphic gauge transformation
 \begin{equation*}
 Y=T\tilde{Y}=\Biggl( T_{0}+\sum_{k=1}^{\infty}T_{k}z^{k}\Biggr)\tilde{Y},
 \end{equation*}
 which transforms it to a simpler equation
 \begin{equation*}
 \frac{{\rm d}\tilde{Y}}{{\rm d}z}=z^{-1}\bigl(J+R_{1}z+\cdots+R_{p}z^{p} \bigr),
 \end{equation*}
 where $J=T_{0}^{-1}\mu T_{0}$ is the Jordan form of the matrix $\mu$ and the entries $(R_{k})_{j}^{i}$ are different from zero only if the eigenvalues $\mu_{l}$ of $\mu$ satisfy $\mu_{i}-\mu_{j}=k\in\mathbb{N}_{> 0}$.

 Writing $\mu_{i}=d_{i}+s_{ii}$, where $d_{i}\in\mathbb{Z}$ and $0\leq \operatorname{Re}(s_{ii})<1$,
 we can write $J=D+S$ with~$D$ a~diagonal matrix with $d_{i}$ as eigenvalues ($S$ is the only part of $J$ which contributes to the mono\-dromy). If we also set $R:=\sum_{k=1}^{p}R_{k}$, then a fundamental matrix solution of equation~(\ref{eq:TheDifferentialEquation}) is\looseness=-1
 \begin{equation}
 \label{eq:LeveltSolution}
 Y_{\rm L}=Tz^{D}z^{R+S}.
 \end{equation}
 This particular kind of solution is called \emph{Levelt fundamental matrix solution}. The monodromy around $z=0$ is the matrix
 \begin{equation*}
 \tilde{M}:={\rm e}^{2\pi {\rm i}(R+S)}.
 \end{equation*}
 In~\cite{Dubrovin:1994hc}, it is shown that this matrix is (locally) constant for all points $p$ of the Dubrovin--Frobenius manifold~$M$.

\section{The exponent of formal monodromy}
 \label{sec:ExponentOfFormalMonodromy}
 In this section, we write a formal solution at $z=\infty$ of equation~(\ref{eq:TheDifferentialEquation}) and compute its exponent of formal monodromy. On Proposition~\ref{prop:ExponentMultiplication}, we show that the exponent of formal monodromy only depends on the natural number $m$ corresponding to the $I_{2}(m)$ $F$-manifold appearing in Hertlings decomposition (see equality~(\ref{eq:ChangeCoordinates})). Recall that, along the caustic, on the orthonormal basis $N$, $f_{i}$, $i=2,\dots, n$ the matrix $U$ of $E\circ$ is diagonal with the first two eigenvalues equal to~$u_{2}$ and the matrix $V$ of $\mu$ is antisymmetric. System~(\ref{eq:TheDifferentialEquation}) reads
 \begin{equation*}
 \frac{{\rm d}Y}{{\rm d}z}=\biggl( \frac{1}{z}V-U\biggr)Y.
 \end{equation*}

 We start by doing a formal Gauge transformation
 \begin{equation}
 \label{eq:FormalGauge}
 Y(z,u)=G\tilde{Y}=\Biggl({\rm Id}+\sum_{k=1}^{\infty}G_{k}(u)z^{-k}\Biggr)\tilde{Y}(z,u)=\Biggl(\sum_{k=0}^{\infty}G_{k}z^{-k}\Biggr)\tilde{Y},
 \end{equation}
 where the matrices $G_{k}$ are to be determined. Since the matrix $U$ has $n-2$ different eigenvalues we wish to find an equivalent block diagonal system that should consist of a $2\times 2$ diagonal block and $n-2$ blocks of dimension $1$. Setting
 \begin{equation}
 \label{eq:BlockDiagonalSystem}
 \frac{{\rm d}\tilde{Y}}{{\rm d}z}=\Biggl(-U+\sum_{k=1}^{\infty}B_{k}(u)z^{-k}\Biggr)\tilde{Y},
 \end{equation}
 where the matrices $B_{k}$ are also to be determined, we get the recursive relations
 \begin{equation}
 \label{eq:GRecursiveRelations}
 -[U,G_{k}]+(k-1)G_{k-1}+VG_{k-1}-\sum_{s=1}^{k-1}G_{k-s}B_{s}=B_{k} \qquad\mathrm{for}\ k\geq 1.
 \end{equation}
 So if we already now $G_{1},\dots,G_{k-1}$ and $B_{1},\dots,B_{k-1}$, we can try to solve the above equation and obtain $G_{k}$ and $B_{k}$. For $k=1$, we need to solve
 \begin{equation*}
 -[U,G_{1}]+V=B_{1}.
 \end{equation*}
 The matrix $B_{1}=V+[G_{1},U]$ must have entries
 \begin{equation*}
 (B_{1})^{i}_{j}=V^{i}_{j}-(G_{1})^{i}_{j}(u_{i}-u_{j}),
 \end{equation*}
 so whenever $1\neq i\neq 2$ or $1\neq j\neq 2$ (recall $u_{1}=u_{2}$), we choose
 \begin{equation}
 \label{eq:G1Equation}
 (G_{1})^{i}_{j}=\frac{(V)^{i}_{j}}{u_{i}-u_{j}},\qquad
 (B_{1})^{i}_{j}=0,
 \end{equation}
 and therefore all entries of $B_{1}$ are zero except for the highest leftmost $2\times 2$ block which is
 \begin{equation*}
 \begin{pmatrix}
 0 & V^{1}_{2}\\
 -V^{1}_{2}& 0
 \end{pmatrix}
 \end{equation*}
 and $(G_{1})^{1}_{2}=(G_{1})^{2}_{1}=(G_{1})^{i}_{i}=0$.

 The equations for $k>1$ can be solved in an analogous way. We obtain that the only non-zero entries of the matrix $B_{k}$ are the ones in the diagonal and the highest leftmost $2\times 2$ block. That is, after the formal Gauge transformation~(\ref{eq:FormalGauge}), we obtain the block diagonal system~(\ref{eq:BlockDiagonalSystem}).

 \begin{Remark}
 Note that the matrices $G_{k}$ are defined uniquely modulo $\operatorname{ker} ({\operatorname{ad}U})$.
 If we set $(G_{k})^{1}_{2},(G_{k})^{2}_{1},(G_{k})^{i}_{i}=0$, then the matrices $G_{k}$ are uniquely defined. We could also do another Gauge transformation $Y=D\bar{Y}$ with $D$ a block-diagonal matrix and still obtain a system with diagonal principal part. In our case, this choice is fixed by writing system~(\ref{eq:TheDifferentialEquation}) in the orthonormal basis $N$, $f_{i}$, $i=2,\dots,n$.
 \end{Remark}
 We now do the gauge transformation $\tilde{Y}={\rm e}^{-Uz}\bar{Y}$, since ${\rm e}^{-Uz}$ acts by scalar multiplication on each block, from system~(\ref{eq:BlockDiagonalSystem}) we obtain a new system of the form
 \begin{equation}
 \label{eq:IntermidateSystem}
 \frac{{\rm d}\bar{Y}}{{\rm d}z}=\Biggl( B_{1}z^{-1}+\sum_{k=2}^{\infty}B_{k}z^{-k}\Biggr)\bar{Y}.
 \end{equation}
 Since all the matrices $B_{k}$ have the same block structure, this last system is a direct sum of Fuchsian systems. Notice that the eigenvalues of $B_{1}$ are zero and $\pm {\rm i}V^{1}_{2}$. We will suppose that $2{\rm i} V^{1}_{2}\notin\mathbb{Z}\setminus \{ 0 \}$ so that the matrix $B_{1}$ is non-resonant. In Proposition~\ref{prop:ExponentMultiplication}, we will show that in our case this always holds true. Since $B_{1}$ is non-resonant, we can find a gauge transformation
\begin{equation*}
 \bar{Y}=H\hat{Y}=H_{0}\Biggl({\rm Id}+\sum_{k=1}^{\infty}H_{k}z^{-k}\Biggr)\hat{Y},
 \end{equation*}
 where the matrices $H_{k}$ have the same block structure as the matrices $B_{k}$ and such that $\bar{Y}$ satisfies the equation
 \begin{equation*}
 \frac{{\rm d}\hat{Y}}{{\rm d}z}=\frac{B}{z}\hat{Y}
 \end{equation*}
 with
 \begin{align*}
 & B:=H_{0}^{-1}B_{1}H_{0}=\operatorname{diag}\bigl({\rm i}V^{1}_{2},-{\rm i}V^{1}_{2},0,\dots,0\bigr),\\
 & \hat{B}_{k}:=H_{0}^{-1}B_{k}H_{0}\qquad \mathrm{for}\ k\geq2,
 \end{align*}
 and
 \begin{equation*}
 [B,H_{k}]+\frac{1}{k}H_{k}=-\hat{B}_{k+1}-\sum_{l=1}^{k-1}\hat{B}_{k+1-l}H_{l} \qquad \mathrm{for}\ k\geq 1.
 \end{equation*}
 Note that after choosing the diagonalizing matrix $H_{0}$ the matrices $H_{k}$ are uniquely determined. Before proceeding to write down the formal solution of equation~(\ref{eq:TheDifferentialEquation}) let us pause a bit to show that we can choose the highest leftmost block of the matrix $H_{0}=H_{0}(u)$ in a special way that will allow us to find isomonodromic fundamental matrix solutions of equation~(\ref{eq:TheDifferentialEquation}).
 \begin{Lemma}
 \label{lem:DiagonalizingV}
 Consider the matrix
 \begin{equation*}
 \tilde{V}=\begin{pmatrix}
 0 & V^{1}_{2}\\
 -V^{1}_{2}&0
 \end{pmatrix}.
 \end{equation*}
 Then there exists a matrix $H_{0}=H_{0}(u)$ which diagonalizes $\tilde{V}$ and such that
 \begin{equation*}
 {\rm d}H_{0}=-\sum_{i=2}^{n}E_{2}\omega_{i}E_{2}{\rm d}u_{i}.
 \end{equation*}
 That is, the columns of $H_{0}$ are $({\rm id}\otimes\pi_{2}\circ)i^{*}\nabla$-flat.
 \end{Lemma}
 \begin{proof}
 Let $\bar{H}_{0}$ be any invertible matrix whose columns are $({\rm id}\otimes\pi_{2}\circ)i^{*}\nabla$-flat (this matrix exists thanks to Proposition~\ref{prop:FlatnessDiagonalConnection}). Note that the connection $({\rm id}\otimes\pi_{2}\circ)i^{*}\nabla$ is a connection on a bundle of rank two and therefore $\bar{H}_{0}$ is a two by two matrix. Let $_{2}\omega=\sum_{i=2}^{n}\,_{2}\omega_{i}{\rm d}u_{i}$ be the connection form of the connection $({\rm id}\otimes\pi_{2}\circ)i^{*}\nabla$, by definition of the matrix $\bar{H}_{0}$, we have
 \begin{equation*}
 {\rm d}\bigl(\bar{H}_{0}^{-1}\tilde{V}\bar{H}_{0}\bigr)=\bar{H}_{0}^{-1}\bigl({\rm d}\tilde{V}+\big[_{2}\omega,\tilde{V}\big]\bigr)\bar{H}_{0}.
 \end{equation*}
 Let us see that ${\rm d}\tilde{V}+\big[_{2}\omega,\tilde{V}\big]=0$. From equation~(\ref{eq:UsefulIdentities2}) we have ${\rm d}V+[\omega,V]=0$. This gives
 \begin{equation*}
 \frac{\partial V^{1}_{2}}{\partial u_{i}}=\sum_{s=1}^{n}V^{1}_{s}(\omega_{i})^{s}_{2}-(\omega_{i})^{1}_{s}V^{s}_{2}.
 \end{equation*}
 From equations~(\ref{eq:UsefulIdentities1}), we get
 \begin{equation*}
 \sum_{s=3}^{n}V^{1}_{s}(\omega_{i})^{s}_{2}-(\omega_{i})^{1}_{s}V^{s}_{2}=V^{1}_{i}(\omega_{i})^{i}_{2}-(\omega_{i})^{1}_{i}V^{i}_{2},
 \end{equation*}
 and from equation~(\ref{eq:UsefulIdentities3}) the above sum vanishes. Therefore,
 \begin{equation*}
 \frac{\partial\tilde{V}^{1}_{2}}{\partial u_{i}}=\frac{\partial V^{1}_{2}}{\partial u_{i}}=[V,\omega_{i}]^{1}_{2}=\big[\tilde{V}, \,_{2}\omega_{i}\big]^{1}_{2}.
 \end{equation*}
 Hence the matrix $\bar{H}_{0}^{-1}\tilde{V}\bar{H}_{0}$ is constant. Therefore, we can find a constant matrix $C$ such that $C^{-1}\bar{H}_{0}^{-1}\tilde{V}\bar{H}_{0}C$ is diagonal. Since $C$ is constant we can take $H_{0}=\bar{H}_{0}C$.
 \end{proof}

 Let us go back to the formal solution of equation~(\ref{eq:TheDifferentialEquation}). Putting together the Gauge transformations $Y=G\tilde{Y}$, $\tilde{Y}={\rm e}^{-Uz}\bar{Y}$ and $\bar{Y}=H\hat{Y}$, we obtain a formal solution to equation~(\ref{eq:TheDifferentialEquation})%
 \begin{equation}
 \label{eq:FormalSolution}
 Y_{F}=GHz^{B}{\rm e}^{-Uz}=\bigl(H_{0}+(H_{0}H_{1}+G_{1}H_{0})z^{-1}+O\bigl(z^{-2}\bigr)\bigr){\rm e}^{-Uz}z^{B}.
 \end{equation}
 The matrix $B$ is called the \emph{exponent of formal monodromy} and by the above, in order to prove that it is constant we only need to show that $V^{1}_{2}$ is a constant. In the following proposition, we compute $V^{1}_{2}$ explicitly.
 \begin{Proposition}
 \label{prop:ExponentMultiplication}
 Let $(M,\circ,e,E,\eta)$ be a Dubrovin--Frobenius manifold with non-empty caustic~$K$ and suppose that for a point $p\in K$ the germ of $M$ at $p$ as an $F$-manifold is isomorphic to $I_{2}(m)\times(A_{1})^{n-2}$ with $m\geq 3$. Then the only non-zero entries of the exponent of formal monodromy are
 \begin{equation*}
 V^{2}_{1}=\frac{{\rm i}}{2}\frac{m-2}{m} \qquad \text{and} \qquad V^{1}_{2}=-V^{2}_{1}=-\frac{{\rm i}}{2}\frac{m-2}{m}.
 \end{equation*}
\end{Proposition}
\begin{proof}
 By Hertling's decomposition, on a neighborhood of a point $p$ of the caustic there exist coordinates $(t,u_{2},\dots,u_{n})$ such that the Euler vector field is written as $E=\tfrac{2}{m}t\partial_{t}+\sum_{s=2}^{n}u_{s}\pi_{s}$.
 We need to compute
 \begin{equation*}
 V_{1}^{2}=\eta(f_{2},\mu N)=-\eta(f_{2},\nabla_{N}E).
 \end{equation*}
 On the basis $\partial_{t},\pi_{2},\dots,\pi_{n}$, the metric $\eta$ takes the form
\begin{equation*}
 \begin{pmatrix}
 \eta_{11}&\eta_{12}&0&\dots&0\\
 \eta_{21}&\eta_{22}&0&\dots&0\\
 0&0&\eta_{33}&\,&\vdots\\
 \vdots& \,&\,&\ddots&\,\\
 0&0&\dots&\,&\eta_{nn}
 \end{pmatrix}.
\end{equation*}
On the caustic $\{ t=0 \}$, we have $\eta_{11}=t^{m-2}\eta_{22}=0$ and the normal to it is $N=-{\rm i}\frac{\sqrt{\eta_{22}}}{\eta_{12}}\partial_{t}+\frac{{\rm i}}{
\sqrt{\eta_{22}}}\pi_{2}$. On the other hand,
\begin{equation*}
 \nabla E=\frac{2}{m}{\rm d}t\otimes\partial_{t}+\sum_{s=2}^{n}{\rm d} u_{s}\otimes\pi_{s} +\frac{2}{m}t\nabla\partial_{t}+\sum_{s=2}^{n}u_{s}\nabla\pi_{s}.
\end{equation*}
Therefore, using the Christoffel $\Gamma_{ij}^{k}$ symbols of the basis $\partial_{t}$, $\pi_{i}$, $i=2,\dots,n$ gives
\begin{align*}
 &\nabla_{\partial_{t}}E=\Biggl( \frac{2}{m}+\frac{2}{m}t\Gamma_{11}^{1}+\sum_{s=2}^{n}u_{s}\Gamma_{1s}^{1}\Biggr)\partial_{t}+\Biggl( \frac{2}{m}t\Gamma_{11}^{2}+\sum_{s=2}^{n}u_{s}\Gamma_{1s}^{2}\Biggr)\pi_{2}+\cdots,\\
 &\nabla_{\pi_{2}}E=\Biggl(\frac{2}{m}t\Gamma_{21}^{1}+\sum_{s=2}^{n}u_{s}\Gamma_{2s}^{1}\Biggr)\partial_{t}+\Biggl( 1+\frac{2}{m}t\Gamma_{21}^{2}+\sum_{s=2}^{n}u_{s}\Gamma_{2s}^{2}\Biggr)\pi_{2}+\cdots.
\end{align*}
With this, we get
\begin{gather*}
 V^{2}_{1}= {\rm i}\frac{m-2}{m}+\frac{2}{m}t\bigg[ \frac{{\rm i}}{\eta_{22}}\bigl(\Gamma_{21}^{1}\eta_{12}+\Gamma_{21}^{2}\eta_{22} \bigr)-\frac{{\rm i}}{\eta_{12}}\bigl(\Gamma_{11}^{1}\eta_{12}+\Gamma_{11}^{2}\eta_{22} \bigr)\bigg] \\
\hphantom{V^{2}_{1}=}{}+\sum_{s=2}^{n}u_{s}\bigg[ \frac{{\rm i}}{\eta_{22}}\bigl(\Gamma_{2s}^{1}\eta_{12}+\Gamma_{2s}^{2}\eta_{22} \bigr)-\frac{{\rm i}}{\eta_{12}}\bigl(\Gamma_{1s}^{1}\eta_{12}+\Gamma_{1s}^{2}\eta_{22} \bigr)\bigg].
\end{gather*}
 Now using the form of the metric and the fact that, on the caustic $\{ t=0 \}$, we have $\eta_{22,s}=0$ for~$s\geq 2$ ($f_{,s}$ denotes the partial derivative of the function $f$ with respect to the $s$-th coordinate), we get
 \begin{align*}
 &\frac{{\rm i}}{\eta_{22}}\bigl(\Gamma_{22}^{1}\eta_{12}+\Gamma_{22}^{2}\eta_{22} \bigr)=\frac{{\rm i}}{2}\frac{\eta_{22,2}}{\eta_{22}},\\
 &-\frac{{\rm i}}{\eta_{12}}\bigl(\Gamma_{12}^{1}\eta_{12}+\Gamma_{12}^{2}\eta_{22} \bigr)=-\frac{{\rm i}}{2}\frac{\eta_{22,1}}{\eta_{12}},
 \end{align*}
 and for $i\geq s$
 \begin{align*}
 &\frac{{\rm i}}{\eta_{22}}\bigl(\Gamma_{2s}^{1}\eta_{12}+\Gamma_{2i}^{2}\eta_{22} \bigr)=\frac{{\rm i}}{2}\frac{\eta_{22,s}}{\eta_{22}},\\
 &-\frac{{\rm i}}{\eta_{12}}\bigl(\Gamma_{1s}^{1}\eta_{12}+\Gamma_{1s}^{2}\eta_{22} \bigr)=-\frac{{\rm i}}{2}\frac{\eta_{12,s}}{\eta_{12}}.
 \end{align*}
 So on the caustic
 \begin{align*}
 V_{1}^{2}&={\rm i}\Bigg[ \frac{m-2}{m}+\frac{1}{2}\Biggl( u_{2}\biggl( \frac{\eta_{22,2}}{\eta_{22}}-\frac{\eta_{22,1}}{\eta_{12}}\biggr)+\sum_{s=3}^{n}u_{s}\biggl( \frac{\eta_{22,s}}{\eta_{22}}-\frac{\eta_{12,s}}{\eta_{12}}\biggr)\Biggr)\Bigg]\\
 &={\rm i}\Bigg[ \frac{m-2}{m}+\frac{1}{2}\Biggl( \sum_{s=2}^{n}u_{s}\biggl( \frac{\eta_{22,s}}{\eta_{22}}-\frac{\eta_{12,s}}{\eta_{12}}\biggr)+\frac{u_{2}}{\eta_{12}}(\eta_{12,2}-\eta_{22,1})\Biggr)\Bigg].
 \end{align*}
 Along the caustic, we have $E=\sum_{i=2}^{n}u_{i}\pi_{i}$ and the condition $\mathcal{L}_{E}\eta=(2-d)\eta$ implies $E(\eta_{22})=-d\eta_{22}$ and $E(\eta_{12})=\bigl(-d+\frac{m-2}{m}\bigr)\eta_{12}$. This gives
 \begin{equation*}
 V^{2}_{1}=\frac{{\rm i}}{2}\biggl( \frac{m-2}{m}+\frac{u_{2}}{\eta_{12}}(\eta_{12,2}-\eta_{22,1})\biggr).
 \end{equation*}
On these coordinates, we also have $\eta(e,-)=\eta_{12}{\rm d}t+\sum_{i=2}^{n}\eta_{ii}{\rm d}u_{i}$ but by Lemma~\ref{prop:Closedness} this form is closed and therefore $\eta_{12,2}-\eta_{22,1}=0$.
\end{proof}

\section{Stokes and connection matrices at the caustic}
\label{sec:StokesMatrices}
In this section, we state Sibuya's theorem asserting that, for all $\nu\in\mathbb{Z}$ there exist appropriate sectors $S_{\nu}$ and holomorphic solutions $Y_{\nu}$ of~(\ref{eq:TheDifferentialEquation}) having asymptotic expansion~(\ref{eq:FormalSolution}) on $S_{\nu}$. Then we show that the Stokes matrices of the fundamental matrix solutions $Y_{\nu}$ are constant. We also show that the connection matrix $C$ relating the matrix $Y_{0}$ with the Levelt fundamental matrix solution $Y_{\rm L}$ is constant and therefore system~(\ref{eq:TheDifferentialEquation}) is strongly isomonodromic.

First, we define the sectors $S_{\nu}$. The gauge transformation~(\ref{eq:FormalGauge}) is usually divergent, but there are certain sectors $S_{\nu}$ of the $z$-plane in which this formal power series is the asymptotic expansion of a holomorphic gauge transformation which takes equation~(\ref{eq:TheDifferentialEquation}) to the block diagonal equation~(\ref{eq:IntermidateSystem}).
\begin{Definition}
 A line $\ell$ through the origin of the $z$-plane is called \emph{admissible} for the system~(\ref{eq:TheDifferentialEquation}) if for all $z\in\ell\setminus\{ 0 \}$, we have that $\mathrm{Re}(z(u_{i}-u_{j}))\neq 0$ whenever $u_{i}-u_{j}\neq 0$. Let $\phi$ be the oriented angle between the positive real axis and an admissible line $\ell$. For $\epsilon$ sufficiently small and $\nu\in\mathbb{Z}$, we define sectors $S_{\nu}$ of opening angle $\pi+2\epsilon$ by
 \begin{equation*}
 S_{0}:=\{z\in\mathbb{C}\mid \arg(z)\in(\phi-\pi-\epsilon,\phi+\epsilon)\},\qquad
 S_{\nu}:={\rm e}^{{\rm i}\nu\pi}S_{0}.
 \end{equation*}
Note that the intersection of two subsequent sectors has opening angle $2\epsilon$.
\end{Definition}

In the following, $u$ denotes a parameter on a small domain $D\subset\mathbb{C}^{n-1}$, for the applications we have in mind $u=(u_{2},\dots,u_{n})$ are the canonical coordinates on a sufficiently small open set of the caustic.
\begin{Theorem}[Sibuya~\cite{Sibuya}]
 Let $A(z,u)=\sum_{k=0}^{\infty}A_{k}(u)z^{-k}$ with $A_{k}\in \operatorname{Mat}_{n}(\mathcal{O}_{\mathbb{C}^{n-1}})$ be holomorphic on $\{ z\in\mathbb{C} \mid |z| \geq N_{0}>0 \}\times\{ |u|\leq\epsilon_{0} \}$ for some $N_{0}\in\mathbb{N}_{>0},\epsilon_{0}\in\mathbb{R}_{+}$ and such that $A_{0}(u)=\Lambda(u)=\Lambda_{1}\oplus\cdots\oplus\Lambda_{s}$ is diagonal with $s\leq n$ distinct eigenvalues \big(each matrix $\Lambda_{i}$ is diagonal $n_{i}\times n_{i}$ matrix with only one eigenvalue and $\sum n_{i}=n$\big). Then, for any proper subsector~$\bar{S}(\alpha,\beta)$ of $S_{\nu}$ there exist positive numbers $N\geq N_{0}$, $\epsilon\leq\epsilon_{0}$ and a matrix $G(z,u)$ with the following properties:
 \begin{enumerate}\itemsep=0pt
 \item[$1.$] $G(z,u)$ is holomorphic in $(z,u)$ for $|z|\geq N$, $z\in\bar{S}(\alpha,\beta)$ and $|u|\leq\epsilon$.
 \item[$2.$] $G(z,u)$ has uniform asymptotic expansion for $|u|\leq\epsilon$ with holomorphic coefficients $G_{k}(u)$,
 \begin{equation*}
 G(z,u)\sim {\rm Id}+\sum_{k=1}^{\infty}G_{k}(u)z^{-k},\qquad z\rightarrow\infty,\quad z\in\bar{S}(\alpha,\beta),
 \end{equation*}
 where the matrices $G_{k}(u)$ are computed from~\eqref{eq:GRecursiveRelations}
 \item[$3.$] The gauge transformation $Y(z,u)=G(z,u)\tilde{Y}(z,u)$ reduces the system $\tfrac{{\rm d}Y}{{\rm d}z}=AY$ to block diagonal form
 \begin{equation*}
 \frac{{\rm d}\tilde{Y}}{{\rm d}z}=\tilde{B}(z,u)\tilde{Y},\qquad \tilde{B}(z,u)=\tilde{B}_{1}(z,u)\oplus\cdots\oplus \tilde{B}_{s}(z,u)
 \end{equation*}
 and $\tilde{B}$ has uniform asymptotic expansion for $|u|\leq\epsilon$ with holomorphic coefficients $B_{k}(u)$
 \begin{equation*}
 \tilde{B}(z,u)\sim\Lambda(u)+\sum_{k=1}^{\infty}B_{k}(u)z^{-k},\qquad z\rightarrow\infty,\quad z\in\bar{S}(\alpha,\beta).
 \end{equation*}
 \end{enumerate}
\end{Theorem}
Now we apply this theorem to the matrix $A=-U+\mu z^{-1}$, $A_{0}=-U$ of system~(\ref{eq:TheDifferentialEquation}) restricted to the caustic $K$, thus we see that the formal gauge transformation of~(\ref{eq:FormalGauge}) is asymptotic, in proper sectors $S_{\nu}$, to a holomorphic gauge transformation $G_{\nu}$ that takes system~(\ref{eq:TheDifferentialEquation}) to the block diagonal form~(\ref{eq:BlockDiagonalSystem}). We obtain holomorphic fundamental matrix solutions of system~(\ref{eq:TheDifferentialEquation}) of the form
\begin{equation}
 \label{eq:CanonicalHolomorphicSolution}
 Y_{\nu}=G_{\nu}Hz^{B}{\rm e}^{-Uz}:=\hat{Y}_{\nu}z^{B}{\rm e}^{-Uz}
\end{equation}
such that in the sector $S_{\nu}$ we have
\begin{equation}
 \label{eq:AsymptoticsOfHolomorphic}
 Y_{\nu}\sim Y_{F}=\bigl(H_{0}+(H_{0}H_{1}+G_{1}H_{0})z^{-1}+O\bigl(z^{-2}\bigr)\bigr){\rm e}^{-Uz}z^{B}.
\end{equation}
Stokes matrices are defined in the usual way. On the overlap of two adjacent sectors $S_{\nu}\cap S_{\nu+1}$, we have that
\begin{equation*}
 Y_{\nu+1}(z;u)=Y_{\nu}(z;u)\mathbb{S}_{\nu}(u).
\end{equation*}
The matrix $\mathbb{S}_{\nu}$ is called \emph{Stokes matrix}. Using the recursive relations of Section~\ref{sec:ExponentOfFormalMonodromy}, we can compute formal solutions of equation~(\ref{eq:TheDifferentialEquation}). Now we proceed to show that if we choose $H_{0}(u)$ in a~particular way, then the corresponding holomorphic solutions are $j^{*}\nabla$-flat and the matrices~$\mathbb{S}_{\nu}(u)$ are independent of the parameter $u$. In the following we denote by $d$ the differential only with respect to the $u$ variable excluding the $z$ variable.
\begin{Lemma}
 Consider $n$ $j^{*}\bar{\nabla}$-flat linearly independent sections $y_{1}(z;u),\dots,y_{n}(z;u)$ of the bundle $(\pi\circ j)^{*}\mathcal{T}_{M}$ and write them down on the basis $N,\pi_{2},\dots,\pi_{n}$. Let $Y$ be the $n\times n$ matrix having the sections $y_{i}$ as columns. Furthermore, let $Y_{F}$ be the formal solution~\eqref{eq:FormalSolution} of equation~\eqref{eq:TheDifferentialEquation} computed with $H_{0}$ as in Lemma~$\ref{lem:DiagonalizingV}$. Then $Y$ satisfies equation~\eqref{eq:TheDifferentialEquation} and on each sector $S_{\nu}$ there exists a constant matrix $C_{\nu}$ such that $Y=Y_{\nu}C_{\nu}$. In particular, $Y_{\nu}$ is $j^{*}\bar{\nabla}$-flat.
\end{Lemma}
\begin{proof}
 The fact that $Y$ satisfies equation~(\ref{eq:TheDifferentialEquation}) is obvious. Therefore, there exists a holomorphic matrix $C_{\nu}=C_{\nu}(u)$ such that $Y=Y_{\nu}C_{\nu}$ so we just need to show that ${\rm d}C_{\nu}=0$. Let $\bar{\omega}_{u}$ be the connection form of the connection $j^{*}\bar{\nabla}$ disregarding the ${\rm d}z$ component. By the definition of $Y$, we have $-\bar{\omega}_{u}={\rm d}Y\cdot Y^{-1}$, so
 \begin{equation*}
 -\omega_{u}-{\rm d}Y_{\nu}\cdot Y_{\nu}^{-1}=Y_{\nu}{\rm d}C_{\nu}\cdot C_{\nu}^{-1}Y_{\nu}^{-1}.
 \end{equation*}
 By Proposition~\ref{prop:ExponentMultiplication}, we have ${\rm d}B=0$. The asymptotic expansion~(\ref{eq:AsymptoticsOfHolomorphic}) in the sector $S_{\nu}$ and the block structure of $B$ and $U$ tell us that
 \begin{align}
 Y_{\nu}{\rm d} C_{\nu}\cdot C_{\nu}^{-1}Y_{\nu}^{-1}& =-\bar{\omega}_{u}-{\rm d}Y_{\nu}\cdot Y_{\nu}^{-1}\nonumber\\
 & \sim-\bar{\omega}_{u} +z{\rm d}U-[{\rm d}U,G_{1}]-{\rm d}H_{0}\cdot H_{0}^{-1}+O\bigl(z^{-1}\bigr).
 \label{eq:lem1}
 \end{align}
Looking at equation~(\ref{eq:ConnectionForm}), we get
\begin{equation*}
 -\bar{\omega}_{u}+z{\rm d}U=-\sum_{i=2}^{n}\omega_{i}{\rm d}u_{i}.
\end{equation*}
Now suppose that $1\neq\alpha\neq 2$ or $1\neq\beta\neq 2$, from equation~(\ref{eq:UsefulIdentities3})
\begin{equation*}
 \sum_{i=2}^{n}(\omega_{i})^{\alpha}_{\beta}{\rm d}u_{i}=-\frac{{\rm d}u_{\alpha}-{\rm d}u_{\beta}}{u_{\alpha}-u_{\beta}}V^{\alpha}_{\beta}.
\end{equation*}
On the other hand, from equation~(\ref{eq:G1Equation}), we get
\begin{equation*}
 [{\rm d}U,G_{1}]^{\alpha}_{\beta}=\frac{{\rm d}u_{\alpha}-{\rm d}u_{\beta}}{u_{\alpha}-u_{\beta}}V^{\alpha}_{\beta}.
\end{equation*}
Substituting these last three equations in~(\ref{eq:lem1}) gives
\begin{equation*}
 Y_{\nu}{\rm d}C_{\nu}\cdot C_{\nu}^{-1}Y_{\nu}^{-1}\sim-\sum_{i=2}^{n}E_{2}\omega_{i}E_{2}{\rm d}u_{i}-{\rm d}H_{0}\cdot H_{0}^{-1},
\end{equation*}
where the matrix $E_{2}$ has entries $(E_{2})^{\alpha}_{\beta}=\delta^{\alpha}_{1}\delta^{1}_{\beta}+\delta^{\alpha}_{2}\delta^{2}_{\beta}$. By Lemma~\ref{lem:DiagonalizingV}, we get that in the sector~$S_{\nu}$
\begin{equation*}
 Y_{\nu}({\rm d}C_{\nu}\cdot C_{\nu})Y_{\nu}^{-1}\sim O\bigl(z^{-1}\bigr).
\end{equation*}
 Let us write
 \begin{equation*}
 Y_{\nu}{\rm d}C_{\nu}\cdot C_{\nu}^{-1}Y_{\nu}^{-1}\sim\sum_{k=1}^{\infty}F_{k}z^{-k}=:F_{\nu}.
 \end{equation*}
 Using~(\ref{eq:CanonicalHolomorphicSolution}) on $\mathcal{S}_{\nu}$, we have
 \begin{equation*}
 {\rm e}^{-zU}z^{B}{\rm d} C_{\nu}\cdot C_{\nu}^{-1}z^{-B}{\rm e}^{Uz}\sim \hat{Y}_{\nu}^{-1}F_{\nu}\hat{Y}_{\nu}.
 \end{equation*}
 Note that since the matrix $\hat{Y}_{\nu}$ is holomorphic on $z=\infty$, the term $\hat{Y}_{\nu}^{-1}F_{\nu}\hat{Y}_{\nu}$ vanishes as $z^{-1}$ when $z\rightarrow\infty$.

 Let us denote $B=\operatorname{diag}(b_{1},\dots,b_{n})=\operatorname{diag}\bigl(\tfrac{m-2}{2m},-\tfrac{m-2}{2m},0,\dots,0\bigr)$ then, since both ${\rm e}^{Uz}$ and $z^{B}$ are diagonal matrices, we have
 \begin{equation}
 \label{eq:ImportantAsymptotics}
 O\bigl(z^{-1}\bigr)\sim\bigl({\rm e}^{-zU}z^{B}{\rm d} C_{\nu}\cdot C_{\nu}^{-1}z^{-B}{\rm e}^{Uz}\bigr)^{\alpha}_{\beta}={\rm e}^{u_{\beta}-u_{\alpha}}z^{b_{\alpha}-b_{\beta}}\bigl({\rm d}C_{\nu}\cdot C_{\nu}^{-1}\bigr)^{\alpha}_{\beta}.
 \end{equation}
 Let $1\neq\alpha\neq 2$ or $1\neq\beta\neq 2$ and $\alpha\neq\beta$, then ${\rm e}^{u_{\beta}-u_{\alpha}}\neq 0$. Since the sector $S_{\nu}$ has opening angle bigger than $\pi$ this sector intersects the line $\operatorname{Re}((u_{\beta}-u_{\alpha})z)=0$. On one side of this line the function ${\rm e}^{(u_{\beta}-u_{\alpha})z}$ diverges when $z\rightarrow\infty$. But the above expression must vanish as $z^{-1}$ when $z\rightarrow\infty$ so whenever $1\neq\alpha\neq 2$ or $1\neq\beta\neq 2$ and $\alpha\neq\beta$ we must have
 \begin{equation*}
 ({\rm d}C_{\nu}\cdot C_{\nu})^{\alpha}_{\beta}=0.
 \end{equation*}
 For $\alpha=\beta$ from~(\ref{eq:ImportantAsymptotics}), we get
 \begin{equation*}
 \bigl({\rm d}C_{\nu}\cdot C_{\nu}^{-1}\bigr)\sim O\bigl(z^{-1}\bigr).
 \end{equation*}
 Since this matrix does not depend on $z$, we again must have
 \begin{equation*}
 ({\rm d}C_{\nu}\cdot C_{\nu})^{\alpha}_{\beta}=0.
 \end{equation*}
 Finally, for $\alpha,\beta\in\{ 1,2 \}$ and $\alpha\neq\beta$ from~(\ref{eq:ImportantAsymptotics}), we get
 \begin{equation*}
 O\bigl(z^{-1}\bigr)\sim z^{\pm\tfrac{m-2}{m}}\bigl({\rm d}C_{\nu}\cdot C_{\nu}^{-1}\bigr)^{\alpha}_{\beta},
 \end{equation*}
 and again we conclude
 \begin{equation*}
 \bigl({\rm d}C_{\nu}\cdot C_{\nu}^{-1}\bigr)^{\alpha}_{\beta}=0.
 \end{equation*}
 Therefore, $\bigl({\rm d}C_{\nu}\cdot C_{\nu}^{-1}\bigr)=0$ and the matrix $C_{\nu}$ is constant.
\end{proof}

\begin{Proposition}
 \label{prop:ConstantStokes}
 The Stokes matrices associated to the formal solution~\eqref{eq:FormalSolution} of equation~\eqref{eq:TheDifferentialEquation} are constant.
\end{Proposition}
\begin{proof}
 By the previous lemma, for all $\nu\in\mathbb{Z}$ we have that ${\rm d}Y_{\nu}=\omega Y_{\nu}$. We also have that $Y_{\nu+1}=Y_{\nu}\mathbb{S}_{\nu}$, so that
 \begin{align*}
 {\rm d}Y_{\nu+1} =\omega Y_{\nu}\mathbb{S}_{\nu}+Y_{\nu}{\rm d}\mathbb{S}_{\nu} =\omega Y_{\nu+1}+Y_{\nu}{\rm d}\mathbb{S}_{\nu}.
 \end{align*}
 Since ${\rm d}Y_{\nu+1}=\omega Y_{\nu+1}$, we conclude ${\rm d}\mathbb{S}_{\nu}=0$ for all $\nu\in\mathbb{Z}$.
\end{proof}
\begin{Proposition}
 \label{ConstantConnectionMatrix}
 Let $C=C(u)$ be the matrix relating the fundamental matrix solution $Y_{\rm L}$ in Levelt form~\eqref{eq:LeveltSolution} around $z=0$ and the fundamental matrix solution $Y_{0}$ having asymptotic expansion~\eqref{eq:FormalSolution} on $S_{0}$. Then ${\rm d}C=0$.
\end{Proposition}
\begin{proof}
 We have that $Y_{\rm L}=Y_{0}C$ and both $Y_{\rm L}$ and $Y_{0}$ satisfy ${\rm d}Y_{*}=\omega Y_{*}$. Hence
 \begin{align*}
 \omega Y_{\rm L} ={\rm d}Y_{0}\cdot C+Y_{0}{\rm d} C =\omega Y_{\rm L}+Y_{0}{\rm d}C.
 \end{align*}
 Hence ${\rm d}C=0$.
\end{proof}

The fact that the monodromy of the Levelt fundamental solution~(\ref{eq:LeveltSolution}) is constant and the Propositions~\ref{prop:ExponentMultiplication},~\ref{prop:ConstantStokes} and~\ref{ConstantConnectionMatrix} imply the following.
\begin{Theorem}
 \label{thm:Isomonodromy}
 Let $(M,\circ,e,E,\eta)$ be a Dubrovin--Frobenius manifold with non-empty caustic~$K$ and suppose that for a point $p\in K$ the germ of $M$ at $p$ as an $F$-manifold is isomorphic to $I_{2}(m)\times(A_{1})^{n-2}$ with $m\geq 3$. Then the fundamental matrix solutions $Y_{\rm L},Y_{\nu},\nu\in\mathbb{Z}$ have constant monodromy data.
 \end{Theorem}

 \section{Three-dimensional examples}
\label{sec:Examples}
In this section, we use Proposition~\ref{prop:ExponentMultiplication} to compute the decomposition $(M,p)\cong I_{2}(m)\times A_{1}$ for the three-dimensional polynomial massive Dubrovin--Frobenius manifolds.

Locally, using flat coordinates $(x,y,z)$, a Dubrovin--Frobenius manifold can be described by means of a single function $F$ called the potential. In the three-dimensional case, if we suppose that $F$ is polynomial and the Dubrovin--Frobenius manifold is massive, there are only three possibilities: the Dubrovin--Frobenius manifolds corresponding to the singularities~$A_{3}$,~$B_{3}$ and~$H_{3}$. The corresponding potentials are (see~\cite[Chapter~1, Example~1.4]{Dubrovin:1994hc})
\begin{gather*}
 F_{A}=\frac{1}{2}x^{2}z+\frac{1}{2}xy^{2}-\frac{1}{16}y^{2}z^{2}+\frac{1}{960}z^{5}, \\
 F_{B}=\frac{1}{2}x^{2}z+\frac{1}{2}xy^{2}+\frac{1}{6}y^{3}z+\frac{1}{6}y^{2}z^{3}+\frac{1}{210}z^{7}, \\
 F_{H}=\frac{1}{2}x^{2}z+\frac{1}{2}xy^{2}+\frac{1}{6}y^{3}z^{2}+\frac{1}{20}y^{2}z^{5}+\frac{1}{3960}z^{11}.
 \end{gather*}
In these coordinates, the metric $\eta$ is given by $\eta_{ij}=\tfrac{\partial F}{\partial t_{1}\partial t_{i}\partial t_{j}}$ (here we identify the indices $x\mapsto 1$, $y\mapsto 2$, $z\mapsto 3$) and in all three cases we get
\begin{equation*}
 \begin{pmatrix}
 0&0&1\\
 0&1&0\\
 1&0&0
 \end{pmatrix}.
\end{equation*}
In these coordinates, the structure constants of the multiplication $\circ$ are $c^{k}_{ij}=\sum_{s}\eta^{ks}\tfrac{\partial F}{\partial t_{s}\partial t_{i}\partial t_{j}}$. The corresponding Euler vector fields are
\begin{align*}
 & E_{A}=x\frac{\partial}{\partial x}+\frac{3}{4}y\frac{\partial}{\partial y}+\frac{1}{2}z\frac{\partial}{\partial z},\\
 & E_{B}=x\frac{\partial}{\partial x}+\frac{2}{3}y\frac{\partial}{\partial y}+\frac{1}{3}z\frac{\partial}{\partial z},\\
 & E_{H}=x\frac{\partial}{\partial x}+\frac{3}{5}y\frac{\partial}{\partial y}+\frac{1}{5}z\frac{\partial}{\partial z}.
\end{align*}

We will explicitly compute $m$ for the $H_{3}$ Dubrovin--Frobenius manifold, the other two cases are similar and much simpler. On the basis $\partial_{x}$, $\partial_{y}$, $\partial_{z}$ the operator of multiplication by the Euler vector field has the form
\begin{equation*}
 \begin{pmatrix}
 x & \frac{7}{10}yz\bigl(2y +z^{3}\bigr) & \frac{1}{20}\bigl(12y^{3}36y^{2}z^{3}+z^{9}\bigr)\vspace{1mm}\\
 \frac{3}{5}y&x+yz^{2}+\frac{1}{5}z^{5}&\frac{7}{10}yz\bigl(2y+z^{3}\bigr)\vspace{1mm} \\
 \frac{1}{5}z^{5}&\frac{3}{5}y&x
 \end{pmatrix}.
\end{equation*}
The discriminant of the characteristic polynomial of this matrix is a multiple of the polynomial
\begin{equation*}
 y^{2}\bigl(y-z^{3}\bigr)^{5}\bigl(27y+5z^{3}\bigr)^{3}.
\end{equation*}
Along this surface (called the \emph{bifurcation diagram}), multiplication by the Euler vector field has a repeated eigenvalue and therefore the caustic is contained in this surface. We can divide this surface into two parts, the caustic and the \emph{semisimple coalescence locus} (for more information about this locus see~\cite{MR4094756}). In order to identify the semisimple coalescence locus, we use the following lemma.
\begin{Lemma}
 Let $(M,\circ,e)$ be an $F$-manifold. Suppose that at a point $p\in M$ there exists a~vector~$v$ such that the operator $v\circ$ has different eigenvalues $u_{i}\neq u_{j}$ if $i\neq j$. Then $p$ is a~semisimple point.
\end{Lemma}
\begin{proof}
 Let $v\circ e_{i}=u_{i}e_{i}$ then $v\circ(e_{i}\circ e_{i})=e_{i}\circ v\circ e_{i}=u_{i}e_{i}\circ e_{i}$ so that $e_{i}\circ e_{i}$ is an eigenvector of $v\circ$ with eigenvalue $u_{i}$. Since all eigenvalues are different, we obtain $e_{i}\circ e_{i}=\lambda_{i}e_{i}$ and $\pi_{i}:=\tfrac{e_{i}}{\lambda_{i}}$ satisfies $\pi_{i}\circ\pi_{i}=\pi_{i}$. Now $u_{i}(\pi_{i}\circ\pi_{j})=v\circ(\pi_{i}\circ\pi_{j})=u_{j}(\pi_{i}\circ\pi_{j})$ but since $u_{i}\neq u_{j}$, we obtain $\pi_{i}\circ\pi_{j}=0$.
\end{proof}

Along the first component of this surface $y=0$, multiplication by $\partial_{y}$ has three different eigenvalues and thus $y=0$ belongs to the semisimple coalescence locus. To identify the caustic note that if a point is semisimple, then the operator of multiplication by any tangent vector is diagonalizable, indeed the orthogonal idempotents are a basis of eigenvectors. Along the components $y=z^{3}$ and $y=-\tfrac{5}{27}z^{3}$, the operator of multiplication by $\partial_{y}$ is not diagonalizable and therefore the caustic is the union of this two components.

The component $y=z^{3}$ is parametrized by $x=r$, $y=s^{3}$, $z=s$ and the tangent space to this surface is generated by $\partial_{r}=e$ and $\partial_{s}=3s^{2}\partial_{y}+\partial_{z}$. In this basis, multiplication by $\partial_{s}$ has matrix
\begin{equation*}
 \begin{pmatrix}
 0&\frac{175}{4}s^{8}\vspace{1mm}\\
 1& 9s^{4}
 \end{pmatrix}.
\end{equation*}
The eigenvectors of this matrix are $e_{2}=-\tfrac{25}{2}s^{4}\partial_{x}+3s^{2}\partial_{y}+\partial_{z}$ and $e_{3}=\frac{7}{2}s^{4}\partial_{x}+3s^{2}\partial_{y}+\partial_{z}$. Along the caustic the tangent space decomposes as the direct sum of a two-dimensional and a~one-dimensional algebras. To identify the unit in each of this algebras we use the Euler vector field. In our previous notation, the eigenvalue associated with $\pi_{2}$ must have multiplicity two and that of $\pi_{3}$ has multiplicity one. Thus, we obtain $e=\pi_{2}+\pi_{3}=-\frac{1}{16s^{2}}e_{2}+\frac{1}{16s^{2}}e_{3}$ so the square lengths of $\pi_{2}$ and $\pi_{3}$ are $-\frac{1}{16s^{4}}$ and $\frac{1}{16s^{4}}$, respectively. The unitary normal is the vector $N=-3s^{2}\partial_{x}+\partial_{y}$ and therefore an orthonormal basis along this component of the caustic consists of the vectors
\[
 N=-3s^{2}\partial_{x}+\partial_{y},\qquad
 f_{2}={\rm i}4s^{2}\pi_{2},\qquad
 f_{3}=4s^{2}\pi_{3}.
\]
On the basis $\partial_{x}$, $\partial_{y}$, $\partial_{z}$, the endomorphism $\mu$ has matrix $\operatorname{diag}\bigl(-\tfrac{2}{5},0,\tfrac{2}{5}\bigr)$ and this gives
\begin{equation*}
 V_{2}^{1}=\eta(N,\mu f_{2})={\rm i}\frac{3}{10}.
\end{equation*}
Therefore, along the component $y=z^{3}$, we have $m=5$. We can parametrize the other component $y=-\tfrac{5}{27}z^{3}$ by $x=r$, $y=-\tfrac{5}{27}s^{3}$, $z=s$. An identical procedure now gives
\begin{equation*}
 m=3.
\end{equation*}
The cases of $B_{3}$ and $A_{3}$ are analogous and simpler. On the $B_{3}$ Dubrovin--Frobenius manifold, the matrix of the endomorphism $\mu$ is $\operatorname{diag}\bigl(-\tfrac{1}{3},0,\tfrac{1}{3}\bigr)$ and the bifurcation diagram has equation
\begin{equation*}
 y^{2}\bigl(2y-3z^{2}\bigr)^{4}\bigl(2y+z^{2}\bigr)^{3}.
\end{equation*}
Again $y=0$ corresponds to the semisimple coalescence locus and the other two components conform the caustic. On the component $\big\{2y-3z^{2}=0 \big\}$, we have $m=4$, and on the component $\big\{2y+z^{2}=0 \big\}$, we have $m=3$. Finally, the $A_{3}$ manifold has bifurcation diagram
\begin{equation*}
 y^{2}\bigl(27y^{2}+8z^{2}\bigr)
\end{equation*}
Once again $y=0$ is the semisimple coalescence locus and on the other component we have $m=3$.

\subsection*{Acknowledgements}
I would like to thank the referees for the useful comments and corrections that helped improve the readability and proofs of this work.

\pdfbookmark[1]{References}{ref}
\LastPageEnding

\end{document}